\pdfoutput=1
\documentclass{article}
\usepackage{arxiv}

\usepackage{graphicx}
\usepackage{hyperref}
\usepackage[utf8]{inputenc}
\usepackage[T1]{fontenc}
\usepackage{url}
\usepackage{amsfonts}
\usepackage{microtype}
\usepackage{mathptmx}
\usepackage{amsmath}
\usepackage{amsthm}
\usepackage{amssymb}
\usepackage{enumerate}

\newtheoremstyle{thm-break}
  {10pt}
  {3pt}
  {\normalfont}
  {}
  {\bfseries}
  {}
  {\newline}
  {\thmname{#1}\thmnumber{ #2}\thmnote{ (#3)}\par\vspace{0.5\baselineskip}}

\theoremstyle{thm-break}

\newtheorem{theorem}{Theorem}[section]
\newtheorem{lemma}[theorem]{Lemma}
\newtheorem{proposition}[theorem]{Proposition}
\newtheorem{corollary}[theorem]{Corollary}
\newtheorem{definition}[theorem]{Definition}
\newtheorem{example}[theorem]{Example}
\newtheorem{remark}[theorem]{Remark}

\newcommand{\QML}{\mathbf{QML}}
\newcommand{\QMLE}{\mathbf{QML}_{\mathrm{eq}}}
\newcommand{\St}{\operatorname{St}}
\newcommand{\Comp}{\operatorname{Comp}}
\newcommand{\Log}{\operatorname{Log}}
\newcommand{\llbracket}{[\![}
\newcommand{\rrbracket}{]\!]}

\title{Component Modalities of Quantum Logic}
\author{Kenji Tokuo\\
Department of Information Engineering, Oita College\\
National Institute of Technology, Japan\\
\texttt{tokuo@oita-ct.ac.jp}}
\date{July 16, 2026}

\begin{document}
\maketitle

\begin{abstract}
This paper determines the structural and proof-theoretic consequences of the forcing condition in relational quantum modal logic, under which every modal transition available at a world is also available at every world compatible with it. We prove that modal successor sets are constant on compatibility components, so boxed truth sets belong to the Boolean algebra of unions of these components. The relation holding exactly between worlds in the same component assigns to each stable proposition its greatest lower and least upper approximations by unions of components, and in hard superselection models these are exactly the approximations by central propositions. We adopt local validity for sequents with multiple conclusions to give a semantics for modal excluded middle on frames with several components. A connectedization obtained by adding one point then shows that component frames, equivalence frames satisfying the forcing condition, and connected compatibility frames with universal modal accessibility have the same logic for sequents with one conclusion. Finally, maximal consistent pairs yield a canonical model, and the calculus obtained by adding T, 4, and B is proved sound and complete for the three frame classes. These results characterize the logical scope of the forcing condition and establish a complete proof theory for component modalities.
\end{abstract}


\section{Introduction}

The lattice approach to quantum logic originates with Birkhoff and von Neumann \cite{BirkhoffVonNeumann1936}. Relational semantics represents quantum negation by a reflexive symmetric relation, and stable subsets of the resulting orthoframe form an ortholattice \cite{DallaChiaraGiuntini2002,DallaChiaraGiuntiniGreechie2004,Goldblatt1974,Mittelstaedt1978}. Modal extensions of quantum and orthomodular logics have been developed through algebraic center operators, relational semantics, and applications to contextuality \cite{DomenechFreytesDeRonde2008,DomenechFreytesDeRonde2009}. Semantics based on compatibility has also been extended to broad classes of nonclassical modal logics by adding an interacting accessibility relation \cite{Holliday2022}. The forcing condition is a special case of this compatibility--accessibility framework. Its componentwise consequences form the subject of the present paper.

Modal orthologic has likewise been developed for epistemic modalities through algebraic semantics and possibility models based on compatibility \cite{HollidayMandelkern2024}. General correspondence theory for modal semantics over graphs studies modal reduction principles across relational environments \cite{ConradieEtAl2024}, while the classical theory of T, 4, B, and equivalence frames provides the standard modal background \cite{BlackburnDeRijkeVenema2001}. The present results connect these lines of research to the forcing condition introduced in earlier work on quantum modal logic \cite{Tokuo2025}.

The forcing condition imposes a substantial structural restriction. Worlds related by compatibility have identical modal successor sets, so accessibility depends only on the compatibility component of the initial world. On a connected projective Hilbert frame, every relation satisfying the forcing condition is state independent. The present paper develops this observation into a structural and proof-theoretic analysis of component modalities.

The first contribution is a classification theorem for relations satisfying the forcing condition. Every such relation is determined by one successor set for each compatibility component. It follows that boxed truth sets are unions of compatibility components and therefore lie in a canonical Boolean subalgebra of the stable set ortholattice. This result supplies the semantic basis for modal excluded middle.

The second contribution concerns the exact component relation, which relates precisely the worlds lying in a common compatibility component. Its box maps a stable proposition to the greatest union of components below it, while its dual diamond gives the least such union above it. In hard superselection models, the component Boolean algebra is the center of the product of sector lattices. The component box and diamond then agree with the corresponding central interior and cover.

The third contribution is a logic coincidence theorem. The logic for sequents with one conclusion over the exact component relation agrees with the logic over equivalence relations satisfying the forcing condition. A connectedization construction further transfers countermodels to connected compatibility frames with universal modal accessibility. This construction adds one compatibility point and preserves the stable operations needed to interpret every formula. This result permits a proof theory for component semantics even though T, 4, and B do not characterize that frame class directly.

The fourth contribution is canonical completeness. The calculus uses sequents with multiple conclusions, while the completeness statement concerns one conclusion. We adopt local validity, under which the selected conclusion may depend on the world. This semantic choice is required by modal excluded middle on frames with several compatibility components. Maximal consistent pairs then provide the saturated canonical worlds needed for quantum negation and modal accessibility. Adding T, 4, and B makes the canonical modal relation an equivalence relation.

The paper is organized as follows. Section~\ref{sec:preliminaries} recalls orthoframe semantics, local sequent validity, and the base calculus. Section~\ref{sec:forcing-structure} proves the structural classification of the forcing condition. Section~\ref{sec:component-modality} studies the component modality and hard superselection models. Section~\ref{sec:logic-coincidence} establishes the coincidence of the three frame logics. Section~\ref{sec:proof-theory} proves soundness and completeness. Section~\ref{sec:comparison} places the results in their algebraic setting and describes the expressive boundary of the forcing condition.

\section{Preliminaries}
\label{sec:preliminaries}

The semantic and proof-theoretic background for the paper is presented here. This section covers stable sets on orthoframes, quantum modal structures, local validity for sequents with multiple conclusions, and the base calculus $\QML$.

\subsection{Stable Sets on Orthoframes}

We begin with the orthoframe semantics of quantum negation and the stable sets that serve as propositions throughout the paper.

\begin{definition}[Orthoframe]
An \emph{orthoframe} is a pair $(W,R_Q)$ such that $W$ is a nonempty set and $R_Q$ is a reflexive and symmetric relation on $W$.
\end{definition}

The relation $R_Q$ is interpreted as compatibility or nonorthogonality. For $X\subseteq W$, define \(X^{\perp}:=\{w\in W\mid R_Q(w,x)\text{ does not hold for any }x\in X\}.\)
Because $R_Q$ is symmetric, the operation $X\mapsto X^{\perp}$ is antitone and satisfies $X\subseteq X^{\perp\perp}$.

\begin{definition}[Stable Set]
Let $(W,R_Q)$ be an orthoframe. A subset $X\subseteq W$ is \emph{stable} if $X=X^{\perp\perp}$. We write $\St(W,R_Q)$ for the family of stable subsets of $W$.
\end{definition}

The family $\St(W,R_Q)$ is a complete ortholattice. Meets are intersections, orthocomplementation is $X\mapsto X^{\perp}$, and joins are given by \(\bigvee_{\lambda\in\Lambda}X_\lambda=\left(\bigcup_{\lambda\in\Lambda}X_\lambda\right)^{\perp\perp}.\)
The bottom and top elements are $\varnothing$ and $W$.

The following closure property will be used to verify the stability of truth sets under conjunction.

\begin{lemma}[Closure under Intersections]
\label{lem:stable-intersection}
Arbitrary intersections of stable sets are stable.
\end{lemma}

\begin{proof}
Let $(X_\lambda)_{\lambda\in\Lambda}$ be a family of stable sets and put $X=\bigcap_{\lambda\in\Lambda}X_\lambda$. Since every set is contained in its double orthocomplement, $X\subseteq X^{\perp\perp}$. To prove the reverse inclusion, fix $w\in X^{\perp\perp}$. The inclusion $X\subseteq X_\lambda$ implies $X_\lambda^{\perp}\subseteq X^{\perp}$. Hence $X^{\perp\perp}\subseteq X_\lambda^{\perp\perp}=X_\lambda$ for every $\lambda\in\Lambda$. Thus $w\in\bigcap_{\lambda\in\Lambda}X_\lambda=X$.
\end{proof}

\subsection{Quantum Modal Structures}

This subsection introduces the language and relational structures used throughout the paper and states the forcing condition connecting quantum compatibility with modal accessibility.

The language contains countably many atomic formulas, conjunction $\wedge$, quantum negation $\neg$, and one box operator $\Box$. Its semantics uses two relations with distinct roles. The relation $R_Q$ represents compatibility or nonorthogonality and determines the interpretation of quantum negation, while $R_M$ represents modal accessibility and determines the interpretation of the box operator. The forcing condition specifies the interaction between these relations. This semantic framework follows the relational tradition for orthologic and can be compared with more general representation frameworks for modal nonclassical logics \cite{Goldblatt1974,Holliday2022}.

Formulas are denoted by $\alpha,\beta,\gamma$. Throughout the paper, $\forall$, $\exists$, $\Longrightarrow$, and $\Longleftrightarrow$ are used in the metalanguage as abbreviations for ``for every,'' ``there exists,'' ``implies,'' and ``if and only if,'' respectively. Worlds are denoted by $u$, $v$, and $w$. Component indices are denoted by $i$, $j$, and $k$, with $I$ as their index set. Stable subsets are denoted by $X$ and $Y$, while $A$ and $B$ denote sets of formulas in the canonical construction. For a binary relation $R$ on a set and an element $u$ of that set, write \(R(u,-):=\{v\mid R(u,v)\}\). We use the abbreviations \(\alpha\vee\beta:=\neg(\neg\alpha\wedge\neg\beta), \Diamond\alpha:=\neg\Box\neg\alpha.\)

\begin{definition}[Forcing Condition]
Let $(W,R_Q)$ be an orthoframe and let $R_M\subseteq W\times W$ be a modal accessibility relation. The relation $R_M$ satisfies the \emph{forcing condition} relative to $R_Q$ if
\[
R_M(u,w)\ \Longrightarrow\ \forall v\in W\,\bigl(R_Q(u,v)\Longrightarrow R_M(v,w)\bigr).
\]
\end{definition}

This condition was introduced in \cite{Tokuo2025}. Combining it with an orthoframe and stable atomic valuations gives the structures used for the semantics.

\begin{definition}[Quantum Modal Structure]
A \emph{quantum modal structure} is a tuple
\[
\mathcal{M}=(W,R_Q,R_M,\rho)
\]
such that $(W,R_Q)$ is an orthoframe, $R_M$ satisfies the forcing condition relative to $R_Q$, and $\rho(p)\in\St(W,R_Q)$ for every atomic formula $p$.
\end{definition}

Truth is defined by
\begin{align*}
w\models p &\quad\Longleftrightarrow\quad w\in\rho(p),\\
w\models\alpha\wedge\beta &\quad\Longleftrightarrow\quad w\models\alpha\text{ and }w\models\beta,\\
w\models\neg\alpha &\quad\Longleftrightarrow\quad
\forall v\,\bigl(R_Q(w,v)\Longrightarrow v\not\models\alpha\bigr),\\
w\models\Box\alpha &\quad\Longleftrightarrow\quad
\forall v\,\bigl(R_M(w,v)\Longrightarrow v\models\alpha\bigr).
\end{align*}
We write $\llbracket\alpha\rrbracket_{\mathcal{M}}$ for the truth set of $\alpha$ in $\mathcal{M}$. When the structure is clear from the context, the subscript is omitted.

The semantic clauses must preserve stability so that every formula continues to denote a proposition in the stable set ortholattice.

\begin{proposition}[Stability of Formula Truth Sets]
\label{prop:formula-stability}
Every formula has a stable truth set in every quantum modal structure.
\end{proposition}

\begin{proof}
Atomic truth sets are stable by definition. Intersections preserve stability by Lemma~\ref{lem:stable-intersection}. The truth set of $\neg\alpha$ is $\llbracket\alpha\rrbracket^{\perp}$, which is stable because $X^{\perp\perp\perp}=X^{\perp}$. It remains to consider $\Box\alpha$.

Put $X=\llbracket\Box\alpha\rrbracket$. The inclusion $X\subseteq X^{\perp\perp}$ holds for every set. Suppose $w\in X^{\perp\perp}$ but $w\notin X$. Then some $u$ satisfies $R_M(w,u)$ and $u\not\models\alpha$. The forcing condition gives $R_M(v,u)$ for every $v$ with $R_Q(w,v)$. Hence every such $v$ fails $\Box\alpha$. Thus every $R_Q$-neighbor of $w$ lies outside $X$, so $w\in X^{\perp}$. Reflexivity of $R_Q$ then gives $w\notin X^{\perp\perp}$, a contradiction. Therefore $w\in X$.
\end{proof}

Having established the semantic closure of the language, we now formulate the corresponding notion of sequent validity.

\subsection{Local Sequent Validity}

This subsection introduces the local interpretation of sequents with multiple conclusions. A sequent is an expression $\Gamma\vdash\Delta$, where $\Gamma$ and $\Delta$ are finite sets of formulas. Under the local interpretation, each world $w$ satisfying every formula in $\Gamma$ satisfies at least one formula in $\Delta$, and the satisfied formula may vary with $w$. The notation $w\models\Gamma$ means that $w$ satisfies every formula in $\Gamma$.

\begin{definition}[Local Validity]
A sequent $\Gamma\vdash\Delta$ is \emph{valid in} $\mathcal{M}$ if, for every $w\in W$,
\[
w\models\Gamma\quad\Longrightarrow\quad
\exists\delta\in\Delta\,\bigl(w\models\delta\bigr).
\]
It is valid in a class of structures if it is valid in every member of that class.
\end{definition}

For an arbitrary set $\Gamma$ of formulas, write $\Gamma\models_{\mathcal{M}}\alpha$ if every world of $\mathcal{M}$ satisfying all formulas in $\Gamma$ also satisfies $\alpha$. For a class $\mathsf{K}$ of structures, write $\Gamma\models_{\mathsf{K}}\alpha$ if $\Gamma\models_{\mathcal{M}}\alpha$ for every $\mathcal{M}\in\mathsf{K}$. When the structure or class is clear, the subscript is omitted. The local interpretation of multiple conclusions allows the selected conclusion to depend on the world. This distinction is essential for the modal excluded middle axiom MEM, $\Gamma\vdash\Box\alpha,\neg\Box\alpha$, displayed with the calculus in Section~\ref{subsec:base-calculus}.

\subsection{Local Consequence for Modal Excluded Middle}

Earlier work formulated consequence for sequents with multiple conclusions by requiring one fixed member of the succedent to follow from the antecedent throughout a model \cite{Tokuo2025}. The soundness argument for modal excluded middle, given at the end of Section~\ref{subsec:base-calculus}, establishes a conclusion selected at each world, and frames with several compatibility components make the distinction substantive. This paper therefore adopts local validity as the formal semantics of the calculus. The following example separates the two consequence relations.

\begin{example}[Modal Alternatives across Components]
Let $W=C_0\sqcup C_1$, where $C_0$ and $C_1$ are distinct connected components of the graph determined by $R_Q$; the formal definition appears in Section~\ref{sec:forcing-structure}. Define $R_C(u,v)$ to hold exactly when $u$ and $v$ belong to the same $R_Q$-component, and put $R_M=R_C$. Put $\rho(p)=C_0$, which is stable because no $R_Q$-edge leaves a component. Then every world in $C_0$ satisfies $\Box p$, while every world in $C_1$ satisfies $\neg\Box p$. Hence the sequent $\varnothing\vdash\Box p,\neg\Box p$ is locally valid. Neither $\Box p$ nor $\neg\Box p$ is valid throughout the model.
\end{example}

\begin{remark}[Revision of Consequence for Multiple Conclusions]
The local clause replaces the interpretation used in the earlier presentation of the calculus $\QML$ \cite{Tokuo2025}, recalled in Section~\ref{subsec:base-calculus}, which requires one fixed conclusion throughout a model. Under that interpretation, the preceding example invalidates MEM. The proof of MEM soundness establishes the local clause adopted here. Both interpretations coincide for sequents with one conclusion, so the consequence relation used in the main completeness theorem is unchanged.
\end{remark}

The example and remark show that MEM expresses pointwise Boolean behavior of boxed formulas. The canonical semantics below is built for this local consequence relation.

\subsection{Base Calculus}
\label{subsec:base-calculus}

The base calculus $\QML$ extends Nishimura's sequent calculus for quantum logic \cite{Nishimura1980} by adding the axiom MEM and the rule K. Its axioms and rules are as follows.

\begin{align*}
&\text{(AX)} && \alpha\vdash\alpha,\\
&\text{(MEM)} && \Gamma\vdash\Box\alpha,\neg\Box\alpha,\\[1mm]
&\text{(WKN)} && \frac{\Gamma\vdash\Delta}{\Pi,\Gamma\vdash\Delta,\Sigma},\\[2mm]
&\text{(CUT)} && \frac{\Gamma_1\vdash\Delta_1,\alpha\qquad
\alpha,\Gamma_2\vdash\Delta_2}
{\Gamma_1,\Gamma_2\vdash\Delta_1,\Delta_2},\\[2mm]
&\text{($\wedge$l$_1$)} && \frac{\alpha,\Gamma\vdash\Delta}{\alpha\wedge\beta,\Gamma\vdash\Delta},\\[2mm]
&\text{($\wedge$l$_2$)} && \frac{\beta,\Gamma\vdash\Delta}{\alpha\wedge\beta,\Gamma\vdash\Delta},\\[2mm]
&\text{($\wedge$r)} && \frac{\Gamma\vdash\Delta,\alpha\qquad\Gamma\vdash\Delta,\beta}
{\Gamma\vdash\Delta,\alpha\wedge\beta},\\[2mm]
&\text{($\neg$l)} && \frac{\Gamma\vdash\Delta,\alpha}{\neg\alpha,\Gamma\vdash\Delta},\\[2mm]
&\text{($\neg$r)} && \frac{\alpha\vdash\Delta}{\neg\Delta\vdash\neg\alpha},\\[2mm]
&\text{($\neg\neg$l)} && \frac{\alpha,\Gamma\vdash\Delta}{\neg\neg\alpha,\Gamma\vdash\Delta},\\[2mm]
&\text{($\neg\neg$r)} && \frac{\Gamma\vdash\Delta,\alpha}{\Gamma\vdash\Delta,\neg\neg\alpha},\\[2mm]
&\text{(K)} && \frac{\Gamma\vdash\alpha}{\Box\Gamma\vdash\Box\alpha}.
\end{align*}
Here $\neg\Delta=\{\neg\delta\mid\delta\in\Delta\}$ and $\Box\Gamma=\{\Box\gamma\mid\gamma\in\Gamma\}$. The modal additions are the axiom MEM and the rule K. The remaining axioms and rules are Nishimura's.

By induction on derivations, every sequent derivable in $\QML$ is locally valid in every quantum modal structure. Appendix~\ref{app:base-soundness} verifies the nonmodal cases. The K rule uses the modal truth clause. For MEM, fix a world $w$. If $w\models\Box\alpha$, the left conclusion holds. If it fails, choose an $R_M$-successor $u$ with $u\not\models\alpha$. The forcing condition supplies the same successor from every $R_Q$-neighbor of $w$, so every such neighbor fails $\Box\alpha$. Hence $w\models\neg\Box\alpha$.

To obtain a calculus for structures in which the modal accessibility relation is an equivalence relation, we add the following axioms to $\QML$:
\begin{align*}
&\text{(T)} && \Box\alpha\vdash\alpha,\\
&\text{(4)} && \Box\alpha\vdash\Box\Box\alpha,\\
&\text{(B)} && \alpha\vdash\Box\Diamond\alpha.
\end{align*}
These axioms correspond to reflexivity, transitivity, and symmetry of the modal accessibility relation, respectively. The resulting calculus is denoted by $\QMLE$.

We write $\Gamma\vdash_{\QML}\Delta$ and $\Gamma\vdash_{\QMLE}\Delta$ when the sequent $\Gamma\vdash\Delta$ is derivable in $\QML$ and $\QMLE$, respectively. For an arbitrary set $\Gamma$ of formulas, $\Gamma\vdash_{\QML}\alpha$ and $\Gamma\vdash_{\QMLE}\alpha$ mean that $\Gamma_0\vdash_{\QML}\alpha$ and $\Gamma_0\vdash_{\QMLE}\alpha$, respectively, for some finite $\Gamma_0\subseteq\Gamma$.

\section{Structure of the Forcing Condition}
\label{sec:forcing-structure}

We now examine the structural consequences of the forcing condition. Modal successor sets are first shown to be constant on compatibility components, after which the Boolean algebra of component propositions and the truth sets of boxed formulas are studied.

\subsection{Component Invariance}

Let $\sim_Q$ be the equivalence relation generated by $R_Q$. Thus $u\sim_Qv$ precisely when there is a finite $R_Q$-path from $u$ to $v$. The equivalence classes of $\sim_Q$ are called the $R_Q$-components, and we write $[u]_Q$ for the component containing $u$. An orthoframe $(W,R_Q)$ is \emph{connected} if it has exactly one $R_Q$-component.

\begin{lemma}[Edge Invariance]
\label{lem:edge-invariance}
Suppose that $R_M$ satisfies the forcing condition relative to $R_Q$. If $R_Q(u,v)$, then
\[
R_M(u,-)=R_M(v,-).
\]
\end{lemma}

\begin{proof}
Let $w\in W$. If $R_M(u,w)$, the forcing condition and $R_Q(u,v)$ give $R_M(v,w)$. Conversely, $R_Q(v,u)$ follows from symmetry, so $R_M(v,w)$ gives $R_M(u,w)$. Thus the successor sets are equal.
\end{proof}

The equality along a single compatibility edge extends along finite paths. This yields the componentwise classification of modal successor sets.

\begin{theorem}[Component Invariance]
\label{thm:component-invariance}
A relation $R_M$ satisfies the forcing condition relative to $R_Q$ if and only if its successor set is constant on every $R_Q$-component. Equivalently, there is a family $(D_C)_{C\in W/{\sim_Q}}$ of subsets of $W$ such that
\[
R_M(u,v)\quad\Longleftrightarrow\quad v\in D_{[u]_Q}.
\]
\end{theorem}

\begin{proof}
Assume the forcing condition. If $u\sim_Qv$, choose a finite path $u=u_0,u_1,\ldots,u_n=v$ with $R_Q(u_k,u_{k+1})$ for every $k<n$. Repeated application of Lemma~\ref{lem:edge-invariance} gives $R_M(u,-)=R_M(v,-)$. Define $D_C=R_M(u,-)$ for any $u\in C$. This is well defined and yields the displayed equivalence.

Assume now the successor set is constant on components. If $R_M(u,w)$ and $R_Q(u,v)$, then $[u]_Q=[v]_Q$. Hence $R_M(v,-)=R_M(u,-)$, so $R_M(v,w)$.
\end{proof}

The connected case is the immediate specialization of the classification, since there is only one compatibility component.

\begin{corollary}[Constant Successors on Connected Frames]
\label{cor:connected-constant}
If $(W,R_Q)$ is connected, every relation satisfying the forcing condition has the form
\[
R_M=W\times D
\]
for a fixed subset $D\subseteq W$.
\end{corollary}

\begin{proof}
There is only one $R_Q$-component, so Theorem~\ref{thm:component-invariance} supplies one successor set $D$ shared by every world.
\end{proof}

\subsection{Component Boolean Algebra}

Let $(C_i)_{i\in I}$ be the family of $R_Q$-components. A subset of $W$ is a \emph{component proposition} if it is a union of members of this family.

\begin{lemma}[Stability of Component Propositions]
\label{lem:component-set-stable}
Every component proposition $X$ is stable, and
\[
X^{\perp}=W\setminus X.
\]
\end{lemma}

\begin{proof}
If $w\in X$, reflexivity gives $R_Q(w,w)$, so $w\notin X^{\perp}$. If $w\notin X$, then $w$ lies in a component disjoint from every component included in $X$. No $R_Q$-edge joins distinct components, hence $w$ is unrelated to every member of $X$, and $w\in X^{\perp}$. Thus $X^{\perp}=W\setminus X$, and double orthocomplementation returns $X$.
\end{proof}

We now use the separation of compatibility components to decompose the entire stable set ortholattice into its component factors.

\begin{proposition}[Product Decomposition]
\label{prop:product-decomposition}
Restriction to components defines an ortholattice isomorphism
\[
\St(W,R_Q)\cong\prod_{i\in I}\St(C_i,R_Q|_{C_i}).
\]
Under this isomorphism, component propositions correspond to tuples whose coordinates are either the bottom or the top of the relevant factor.
\end{proposition}

\begin{proof}
Let $X\in\St(W,R_Q)$. Since no $R_Q$-edge joins different components, $X\cap C_i$ is stable in the restricted frame $C_i$. This defines a map \(\Phi(X)=(X\cap C_i)_{i\in I}.\)
Assume now $X_i$ is stable in $C_i$ for every $i\in I$, and put $X=\bigcup_{i\in I}X_i$. Orthocomplementation is computed componentwise because all $R_Q$-neighbors of a point lie in its own component. Hence $X^{\perp\perp}\cap C_i=X_i^{\perp\perp}=X_i$, so $X$ is stable. This gives the inverse map. Intersections and orthocomplements are preserved componentwise, and therefore joins are preserved as well.

A component proposition has intersection $C_i$ or $\varnothing$ with each component. The converse is immediate.
\end{proof}

The stable sets represented by tuples of bottom and top coordinates are exactly the unions of components, and these form a complete Boolean algebra.

\begin{definition}[Component Boolean Algebra]
The \emph{component Boolean algebra} is
\[
\mathcal{B}_{\Comp}:=\left\{\bigcup_{i\in J}C_i\ \middle|\ J\subseteq I\right\}.
\]
\end{definition}

We next compare this Boolean algebra with the center of the stable set ortholattice.

For a bounded ortholattice $L$, the interval $[0,z]$ is equipped with the relative orthocomplement $x\mapsto x^{\perp}\wedge z$. An element $z\in L$ is \emph{central} if the map \(x\longmapsto (x\wedge z,x\wedge z^{\perp})\) is an ortholattice isomorphism from $L$ onto the product of the intervals $[0,z]$ and $[0,z^{\perp}]$. The \emph{center} of $L$, denoted by $Z(L)$, is the set of all central elements of $L$.

\begin{lemma}[Center of a Direct Product]
\label{lem:product-center}
For every family $(L_i)_{i\in I}$ of bounded ortholattices,
\[
Z\left(\prod_{i\in I}L_i\right)=\prod_{i\in I}Z(L_i).
\]
\end{lemma}

\begin{proof}
Write $0_i$ and $1_i$ for the bottom and top elements of $L_i$. Put $L=\prod_{i\in I}L_i$ and let $z=(z_i)_{i\in I}$ be central in $L$. For each $i\in I$ and $a\in L_i$, let $\widehat a^{\,i}$ denote the tuple whose $i$th coordinate is $a$ and whose remaining coordinates are bottom elements. The global decomposition determined by $z$ sends $\widehat a^{\,i}$ to a pair whose $i$th coordinates are $a\wedge z_i$ and $a\wedge z_i^{\perp}$ and whose remaining coordinates are bottom elements. Its injectivity therefore implies that
\[
a\longmapsto(a\wedge z_i,a\wedge z_i^{\perp})
\]
is injective for every $i\in I$.

Fix $k\in I$. For $b\leq z_k$ and $c\leq z_k^{\perp}$, form elements of $[0,z]$ and $[0,z^{\perp}]$ supported only at coordinate $k$, with respective values $b$ and $c$. Surjectivity of the global decomposition gives $x\in L$ whose image is this pair. The $k$th coordinate of $x$ maps to $(b,c)$. At every coordinate $i\neq k$, the coordinate map sends both $x_i$ and $0_i$ to $(0_i,0_i)$. Its injectivity gives $x_i=0_i$. Thus the coordinate map at $k$ is surjective. Since $k$ was arbitrary, every coordinate map is bijective.

It remains to verify preservation of the relative orthocomplements. Fix $k\in I$ and $a\in L_k$. The orthocomplement of $\widehat a^{\,k}$ in $L$ has $k$th coordinate $a^{\perp}$ and every other coordinate equal to the corresponding top element. Since the global decomposition is an ortholattice isomorphism, it preserves orthocomplements in the product of the intervals. Comparing the $k$th coordinate in the first interval gives
\[
a^{\perp}\wedge z_k=(a\wedge z_k)^{\perp}\wedge z_k,
\]
and comparison in the second interval gives
\[
a^{\perp}\wedge z_k^{\perp}=(a\wedge z_k^{\perp})^{\perp}\wedge z_k^{\perp}.
\]
The coordinate decomposition preserves meets because meet is associative, commutative, and idempotent. Hence it is an ortholattice isomorphism, and $z_k$ is central in $L_k$. Since $k$ was arbitrary, $z_i$ is central in $L_i$ for every $i\in I$.

Conversely, suppose $z_i$ is central in $L_i$ for every $i\in I$. The product of the coordinate ortholattice isomorphisms
\[
a_i\longmapsto(a_i\wedge z_i,a_i\wedge z_i^{\perp})
\]
is an ortholattice isomorphism from $L$ onto $[0,z]\times[0,z^{\perp}]$. Hence $z$ is central in $L$.
\end{proof}

The product lemma gives the precise relation between component propositions and the full center.

\begin{proposition}[Centrality of Component Propositions]
\label{prop:component-center}
The algebra $\mathcal{B}_{\Comp}$ is a complete Boolean subalgebra of $\St(W,R_Q)$, and every member of $\mathcal{B}_{\Comp}$ is central. If every factor $\St(C_i,R_Q|_{C_i})$ has center $\{\varnothing,C_i\}$, then
\[
Z(\St(W,R_Q))=\mathcal{B}_{\Comp}.
\]
\end{proposition}

\begin{proof}
Lemma~\ref{lem:component-set-stable} shows that orthocomplementation in $\mathcal{B}_{\Comp}$ is ordinary set complement, and unions and intersections are computed set-theoretically. Hence $\mathcal{B}_{\Comp}$ is a complete Boolean subalgebra.

Under the product isomorphism of Proposition~\ref{prop:product-decomposition}, a member $X$ of $\mathcal{B}_{\Comp}$ is represented by a tuple whose $i$th coordinate is $C_i$ when $C_i\subseteq X$ and $\varnothing$ otherwise. Every such coordinate is central in its factor, so Lemma~\ref{lem:product-center} shows that $X$ is central.

If every factor has center $\{\varnothing,C_i\}$, Lemma~\ref{lem:product-center} shows that every central element of the product has coordinate $\varnothing$ or $C_i$ in each factor. Under Proposition~\ref{prop:product-decomposition}, these tuples correspond exactly to the component propositions of $W$. Hence
\[
Z(\St(W,R_Q))=\mathcal{B}_{\Comp}.
\]
\end{proof}

\subsection{Modal Truth Sets on Components}

\begin{theorem}[Boolean Centrality of Boxed Formulas]
\label{thm:boxed-centrality}
In every quantum modal structure, the truth set of $\Box\alpha$ belongs to $\mathcal{B}_{\Comp}$.
\end{theorem}

\begin{proof}
Let $u\sim_Qv$. Theorem~\ref{thm:component-invariance} gives $R_M(u,-)=R_M(v,-)$. The truth clause for $\Box$ therefore yields \(u\models\Box\alpha\quad\Longleftrightarrow\quad v\models\Box\alpha.\)
Thus the truth set is constant on components and is a union of components.
\end{proof}

The componentwise constancy of boxed truth sets immediately yields the local form of modal excluded middle.

\begin{corollary}[Local Modal Excluded Middle]
\label{cor:local-mem}
Every world satisfies at least one of $\Box\alpha$ and $\neg\Box\alpha$.
\end{corollary}

\begin{proof}
Put $X=\llbracket\Box\alpha\rrbracket$. By Theorem~\ref{thm:boxed-centrality}, $X$ is a union of components. Lemma~\ref{lem:component-set-stable} gives $X^{\perp}=W\setminus X$. Since $\llbracket\neg\Box\alpha\rrbracket=X^{\perp}$, every world belongs to one of the two truth sets.
\end{proof}

Corollary~\ref{cor:local-mem} gives the semantic content of MEM: the truth set of every boxed formula is a component proposition, whose orthocomplement is its set complement. The conclusion selected by local validity may vary from one world to another.

\section{Component Modality}
\label{sec:component-modality}

The preceding structural results are next applied to the component relation. Its box and diamond are characterized as lower and upper component approximations, and their role in hard superselection models is then established.

\subsection{Semantic Component Operators}

\begin{definition}[Component Relation]
The \emph{component relation} $R_C$ on an orthoframe is defined by
\[
R_C(u,v)\quad\Longleftrightarrow\quad u\sim_Qv.
\]
A quantum modal structure with $R_M=R_C$ is called a \emph{component model}.
\end{definition}

The relation $R_C$ is an equivalence relation and satisfies the forcing condition relative to $R_Q$. For a stable set $X$, define \(\Box_CX:=\{u\in W\mid R_C(u,-)\subseteq X\}.\)
The semantic diamond is \(\Diamond_CX:=(\Box_C(X^{\perp}))^{\perp}.\)

\begin{theorem}[Component Approximation Operators]
\label{thm:component-interior-cover}
Let $(C_i)_{i\in I}$ be the $R_Q$-components and let $X\in\St(W,R_Q)$. Then
\[
\Box_CX=\bigcup\{C_i\mid C_i\subseteq X\}
\]
and
\[
\Diamond_CX=\bigcup\{C_i\mid C_i\cap X\neq\varnothing\}.
\]
Consequently, $\Box_CX$ is the greatest member of $\mathcal{B}_{\Comp}$ below $X$, and $\Diamond_CX$ is the least member of $\mathcal{B}_{\Comp}$ above $X$.
\end{theorem}

\begin{proof}
A world $u$ belongs to $\Box_CX$ exactly when every world in $[u]_Q$ belongs to $X$. This is equivalent to $[u]_Q\subseteq X$, which proves the first equality.

For the second equality, observe that a component $C_i$ is contained in $X^{\perp}$ exactly when $C_i\cap X=\varnothing$. One direction follows from reflexivity. For the other direction, if $C_i$ is disjoint from $X$, then every point of $X$ lies in a component different from $C_i$, so no point of $C_i$ is $R_Q$-related to a point of $X$. Thus \(\Box_C(X^{\perp})=\bigcup\{C_i\mid C_i\cap X=\varnothing\}.\)
This set is a union of components, so its orthocomplement is its set complement by Lemma~\ref{lem:component-set-stable}. The complement is the union of the components meeting $X$.

The first displayed union is a union of components, lies below $X$, and contains every component proposition contained in $X$. The second is a union of components, contains $X$, and is contained in every component proposition containing $X$.
\end{proof}

The approximation characterization gives the standard interior and closure laws of the component operators.

\begin{corollary}[Modal Closure Laws]
\label{cor:component-modal-laws}
The component operators satisfy the following equations on stable sets:
\begin{align*}
\Box_CX&\subseteq X, & \Box_C\Box_CX&=\Box_CX,\\
X&\subseteq\Diamond_CX, & \Diamond_C\Diamond_CX&=\Diamond_CX,\\
\Box_C(X\cap Y)&=\Box_CX\cap\Box_CY, &
\Diamond_C(X\vee Y)&=\Diamond_CX\vee\Diamond_CY.
\end{align*}
Here $\vee$ denotes the join in $\St(W,R_Q)$. Their common range is $\mathcal{B}_{\Comp}$.
\end{corollary}

\begin{proof}
The interior and closure equations follow from the greatest lower approximation and least upper approximation properties in Theorem~\ref{thm:component-interior-cover}. The meet equation follows because a component is contained in $X\cap Y$ exactly when it is contained in both $X$ and $Y$. The join equation follows by orthocomplement duality. A stable set is fixed by either operator exactly when it is a union of components.
\end{proof}

\subsection{Hard Superselection Models}

Let $(\mathcal{H}_i)_{i\in I}$ be a family of nonzero complex Hilbert spaces. For $i\in I$ and every nonzero $x\in\mathcal{H}_i$, let $[x]$ denote the one-dimensional subspace spanned by $x$. Let
$W=\bigsqcup_{i\in I}\mathbb{P}(\mathcal{H}_i)$
be the disjoint union of their projective spaces, where $\mathbb{P}(\mathcal{H}_i)$ is the set of one-dimensional subspaces of $\mathcal{H}_i$. The relation $R_Q$ is defined by
\[
R_Q([x],[y])\quad\Longleftrightarrow\quad [x],[y]\in\mathbb{P}(\mathcal{H}_i)\text{ for some }i\in I\text{ and }\langle x,y\rangle\neq0.
\]
The state space contains the rays internal to each sector and omits rays that superpose different sectors. This choice gives the hard superselection model used below.

\begin{proposition}[Sector Components]
\label{prop:hilbert-components}
The $R_Q$-components of the preceding frame are exactly the sets $\mathbb{P}(\mathcal{H}_i)$.
\end{proposition}

\begin{proof}
Different sectors are $R_Q$-disconnected by definition. Fix one sector. If two rays are nonorthogonal, they are adjacent. If rays $[x]$ and $[y]$ are orthogonal, the ray $[x+y]$ is well defined and is nonorthogonal to both. Hence every pair of rays in the sector is joined by a path of length at most two.
\end{proof}

Having established the compatibility components, we can now apply the general product decomposition to the sector model.

\begin{proposition}[Sector Product Representation]
\label{prop:hilbert-stable-product}
The stable set ortholattice of the hard superselection frame is isomorphic to
\[
\prod_{i\in I}\mathcal{L}(\mathcal{H}_i),
\]
where $\mathcal{L}(\mathcal{H}_i)$ is the ortholattice of closed subspaces of $\mathcal{H}_i$. Under this isomorphism, $\mathcal{B}_{\Comp}$ is the center of the product.
\end{proposition}

\begin{proof}
For one projective Hilbert frame, stable subsets are exactly the sets of rays contained in closed subspaces. Proposition~\ref{prop:product-decomposition} therefore gives the displayed product. The center of each factor $\mathcal{L}(\mathcal{H}_i)$ consists of its bottom and top elements. In an orthomodular lattice, a central element is compatible with every element \cite{Kalmbach1983}. For the Hilbert lattice, compatibility of closed subspaces is equivalent to commutativity of their orthogonal projections. Hence the projection onto a central subspace commutes with every orthogonal projection on $\mathcal{H}_i$. Such a projection belongs to the center of $\mathcal{B}(\mathcal{H}_i)$, the algebra of bounded linear operators on $\mathcal{H}_i$, and is therefore a scalar multiple of the identity. Since it is a projection, it is either zero or the identity. The center of a direct product is the product of the centers of the factors by Lemma~\ref{lem:product-center}. Hence each central element is determined by an independent choice, for every $i\in I$, between the zero subspace and the whole sector. These are exactly the component propositions.
\end{proof}

The following example illustrates the lower and upper component approximations in the simplest nontrivial sector decomposition.

\begin{example}[Two Superselection Sectors]
Let $W=W_0\sqcup W_1$ be the projective spaces of two sectors. Let $X$ be the set of rays in a proper nonzero closed subspace of the first sector. Then \(\Box_CX=\varnothing\) and \(\Diamond_CX=W_0\). The set $X$ contains no complete sector, and its central cover is the first sector. If $Y=W_0$, then $\Box_CY=Y=\Diamond_CY$.
\end{example}

\section{Coincidence of Three Frame Logics}
\label{sec:logic-coincidence}

Attention now turns to the relation among three classes of structures. Restriction to a modal equivalence class and connectedization by an added point yield their coincidence for sequents with one conclusion.

\subsection{Frame Classes}

\begin{definition}[Semantic Frame Classes]
We consider the following classes of quantum modal structures.

\begin{itemize}
\item $\mathsf{C}$ is the class of component models, in which $R_M=R_C$.
\item $\mathsf{E}$ is the class of structures in which $R_M$ is an equivalence relation.
\item $\mathsf{U}$ is the class of structures in which $R_Q$ is connected and $R_M=W\times W$.
\end{itemize}
\end{definition}

Every member of $\mathsf{C}$ belongs to $\mathsf{E}$, and every member of $\mathsf{U}$ belongs to $\mathsf{C}$. We shall prove the converse inclusions at the level of validity for sequents with one conclusion.

The three classes already coincide on connected compatibility frames. Corollary~\ref{cor:connected-constant} gives $R_M=W\times D$ for every relation satisfying the forcing condition on such a frame. If $R_M$ is reflexive, then every world belongs to $D$, so $D=W$ and $R_M=W\times W$. The proof below removes the remaining freedom in a nonconnected countermodel by restricting it to one modal equivalence class and then applying the connectedization construction to its compatibility frame.

\subsection{Restriction to a Modal Equivalence Class}

\begin{lemma}[Containment of Compatibility]
\label{lem:rq-contained-rm}
If $R_M$ satisfies the forcing condition and is reflexive and symmetric, then
\[
R_Q\subseteq R_M.
\]
\end{lemma}

\begin{proof}
Assume $R_Q(u,v)$. Reflexivity gives $R_M(u,u)$. The forcing condition gives $R_M(v,u)$, and symmetry gives $R_M(u,v)$.
\end{proof}

This containment ensures that an equivalence class of the modal relation is also closed under compatibility, which makes restriction possible.

\begin{lemma}[Class Restriction]
\label{lem:class-restriction}
Let $\mathcal{M}=(W,R_Q,R_M,\rho)$ belong to $\mathsf{E}$, let $w\in W$, and let $E$ be the $R_M$-equivalence class of $w$. Define
\[
\mathcal{M}|_E=(E,R_Q|_E,E\times E,\rho_E),
\rho_E(p)=\rho(p)\cap E.
\]
Then $\rho_E(p)$ is stable in the restricted orthoframe, and for every formula $\alpha$ and every $v\in E$,
\[
v\models_{\mathcal{M}}\alpha
\quad\Longleftrightarrow\quad
v\models_{\mathcal{M}|_E}\alpha.
\]
\end{lemma}

\begin{proof}
Lemma~\ref{lem:rq-contained-rm} shows that no world in $E$ is $R_Q$-related to any world in $W\setminus E$. For $Y\subseteq E$, let $Y^{\perp_E}$ denote its orthocomplement in the restricted orthoframe $(E,R_Q|_E)$. Then $Y^{\perp_E}=Y^{\perp}\cap E$. If $X$ is stable in $W$, the absence of edges from $E$ to $W\setminus E$ also gives $(X\cap E)^{\perp_E}=X^{\perp}\cap E$. Applying the same identity once more yields $(X\cap E)^{\perp_E\perp_E}=X^{\perp\perp}\cap E=X\cap E$. Hence $X\cap E$ is stable in the restricted frame.

We prove truth preservation by induction on $\alpha$. The atomic and conjunction cases are immediate. The negation case follows because every $R_Q$-neighbor of a member of $E$ also lies in $E$. For the modal case, the $R_M$-successors of every $v\in E$ in the original model are exactly the members of $E$, because $E$ is its equivalence class. The restricted modal relation is also $E\times E$. Thus the two box clauses quantify over the same set.
\end{proof}

\subsection{Connectedization}

The next construction turns an arbitrary orthoframe into a connected orthoframe while preserving its stable set operations through an embedding.

\begin{definition}[Connectedization by an Added Point]
Let $\mathcal{F}=(W,R_Q)$ be an orthoframe and let $*$ be a new point. Put $W^+=W\cup\{*\}$ and
\[
R_Q^+=R_Q\cup\{(*,x),(x,*):x\in W^+\}.
\]
For a stable set $X\subseteq W$, define
\[
e(X)=
\begin{cases}
W^+,&X=W,\\
X,&X\neq W.
\end{cases}
\]
\end{definition}

The next lemma controls the nontrivial case needed to prove preservation of orthocomplementation.

\begin{lemma}[Nontrivial Orthocomplements]
\label{lem:proper-stable-complement}
If $X$ is a nonempty proper stable subset of an orthoframe, then $X^{\perp}$ is also nonempty and proper.
\end{lemma}

\begin{proof}
If $X^{\perp}=\varnothing$, then stability gives \(X=X^{\perp\perp}=\varnothing^{\perp}=W,\) contrary to properness. If $X^{\perp}=W$, then \(X=X^{\perp\perp}=W^{\perp}=\varnothing,\) contrary to nonemptiness.
\end{proof}

We can therefore verify that the connectedization preserves the stable operations used by the language.

\begin{lemma}[Embedding of Stable Sets]
\label{lem:stable-embedding}
The map $e:\St(W,R_Q)\to\St(W^+,R_Q^+)$ satisfies
\[
e(X\cap Y)=e(X)\cap e(Y)
\]
and
\[
e(X^{\perp})=e(X)^{\perp_+},
\]
where $\perp_+$ denotes orthocomplementation with respect to $R_Q^+$.
\end{lemma}

\begin{proof}
The definition shows that $e(W)=W^+$ and that $e(X)=X$ for every proper stable set $X$. Hence $e$ is injective. The meet equation follows by cases. If $X=Y=W$, both sides are $W^+$. If at least one is proper, the added point belongs to neither side, and the equation reduces to the ordinary intersection identity on $W$.

We prove the orthocomplement equation. If $X=\varnothing$, then $e(X)=\varnothing$ and \(e(X)^{\perp_+}=W^+=e(W)=e(X^{\perp}).\)
If $X=W$, then $e(X)=W^+$ and \(e(X)^{\perp_+}=\varnothing=e(\varnothing)=e(X^{\perp}).\)
Suppose that $X$ is nonempty and proper. The added point is related to every member of $X$ and therefore lies outside $e(X)^{\perp_+}$. For $w\in W$, the added endpoint $*$ lies outside $e(X)=X$, so the new edge incident to $w$ imposes no additional constraint in the orthocomplement clause. Therefore \(e(X)^{\perp_+}=X^{\perp}.\)
By Lemma~\ref{lem:proper-stable-complement}, $X^{\perp}$ is nonempty and proper, so $e(X^{\perp})=X^{\perp}$. This proves the equation in every case. Since $e$ preserves orthocomplementation and $X=X^{\perp\perp}$, every $e(X)$ is stable.
\end{proof}

Having established preservation of meet and orthocomplementation, we extend the embedding to the interpretation of every formula.

\begin{lemma}[Truth Preservation under Connectedization]
\label{lem:truth-connectedization}
Let
\[
\mathcal{N}=(W,R_Q,W\times W,\rho)
\]
be a quantum modal structure with universal modal relation. Define
\[
\mathcal{N}^+=(W^+,R_Q^+,W^+\times W^+,\rho^+),\qquad \rho^+(p)=e(\rho(p)).
\]
Then $R_Q^+$ is connected, $\mathcal{N}^+\in\mathsf{U}$, and
\[
\llbracket\alpha\rrbracket_{\mathcal{N}^+}
=e(\llbracket\alpha\rrbracket_{\mathcal{N}})
\]
for every formula $\alpha$.
\end{lemma}

\begin{proof}
The definition of $R_Q^+$ makes it reflexive and symmetric. The point $*$ is adjacent to every world, so $R_Q^+$ is connected. Lemma~\ref{lem:stable-embedding} shows that every $\rho^+(p)$ is stable. The modal relation is universal and therefore satisfies the forcing condition. Hence $\mathcal{N}^+$ is a quantum modal structure in $\mathsf{U}$.

We prove the truth set equation by induction. The atomic case is the definition of $\rho^+$. The conjunction and negation cases follow from Lemma~\ref{lem:stable-embedding}.

Let $X=\llbracket\alpha\rrbracket_{\mathcal{N}}$. Since the modal relation is universal, $\llbracket\Box\alpha\rrbracket_{\mathcal{N}}$ is $W$ when $X=W$ and is $\varnothing$ otherwise. By the induction hypothesis, the truth set of $\alpha$ in $\mathcal{N}^+$ is $e(X)$. The definition of $e$ gives $e(X)=W^+$ exactly when $X=W$. It follows that $\llbracket\Box\alpha\rrbracket_{\mathcal{N}^+}$ is $W^+$ when $X=W$ and is $\varnothing$ otherwise. Hence $\llbracket\Box\alpha\rrbracket_{\mathcal{N}^+}=e(\llbracket\Box\alpha\rrbracket_{\mathcal{N}})$.
\end{proof}

This truth preservation completes the transfer of countermodels required for the coincidence theorem.

For a class $\mathsf{K}$ of structures, let $\Log(\mathsf{K})$ denote the set of sequents with one conclusion that are valid in every member of $\mathsf{K}$.

\begin{theorem}[Logic Coincidence]
\label{thm:logic-coincidence}
For every set $\Gamma$ of formulas and every formula $\alpha$, the following are equivalent:
\begin{enumerate}[(i)]
\item $\Gamma\models_{\mathsf{E}}\alpha$.
\item $\Gamma\models_{\mathsf{C}}\alpha$.
\item $\Gamma\models_{\mathsf{U}}\alpha$.
\end{enumerate}
Consequently,
\[
\Log(\mathsf{E})=\Log(\mathsf{C})=\Log(\mathsf{U}).
\]
\end{theorem}

\begin{proof}
The inclusions $\mathsf{U}\subseteq\mathsf{C}\subseteq\mathsf{E}$ give \(\Log(\mathsf{E})\subseteq\Log(\mathsf{C})\subseteq\Log(\mathsf{U}).\)
It remains to show that every countermodel in $\mathsf{E}$ yields a countermodel in $\mathsf{U}$.

Suppose $\Gamma\not\models_{\mathsf{E}}\alpha$. Choose $\mathcal{M}\in\mathsf{E}$ and $w$ such that $w\models\Gamma$ and $w\not\models\alpha$. Let $E$ be the $R_M$-equivalence class of $w$. Lemma~\ref{lem:class-restriction} produces a model $\mathcal{M}|_E$ with universal modal relation in which the same truth and falsity statements hold at $w$. Apply Lemma~\ref{lem:truth-connectedization} to $\mathcal{M}|_E$. The resulting model belongs to $\mathsf{U}$ and preserves the truth of every formula at every old world. Hence it is a countermodel to $\Gamma\models_{\mathsf{U}}\alpha$.
\end{proof}

The following remark states a limitation of the argument.

\begin{remark}[Restriction to One Modality]
With several modal relations, restriction to one equivalence class may remove successors for another modality, and formulas connecting different modalities can distinguish exact component relations from coarser equivalence relations.
\end{remark}

\section{Proof Theory of the Component Modality}
\label{sec:proof-theory}

The proof theory of the component modality is developed in this section. After the soundness of T, 4, and B is established, maximal consistent pairs are used to construct the canonical model and prove completeness for the three frame classes.

\subsection{Modal Axioms T, 4, B}

\begin{theorem}[Soundness]
\label{thm:soundness-eq}
Every sequent derivable in $\QMLE$ is locally valid in every structure in $\mathsf{E}$. In particular, if $\Gamma\vdash_{\QMLE}\alpha$, then $\Gamma\models_{\mathsf{E}}\alpha$.
\end{theorem}

\begin{proof}
The axioms and rules of the base calculus preserve local validity. We verify the three added axioms.

For T, suppose $w\models\Box\alpha$. Reflexivity of $R_M$ gives $R_M(w,w)$, so $w\models\alpha$.

For 4, suppose $w\models\Box\alpha$ and let $v$ satisfy $R_M(w,v)$. To show $v\models\Box\alpha$, let $u$ satisfy $R_M(v,u)$. Transitivity gives $R_M(w,u)$, hence $u\models\alpha$. Therefore $w\models\Box\Box\alpha$.

For B, suppose $w\models\alpha$. Let $v$ satisfy $R_M(w,v)$. We prove $v\models\Diamond\alpha$, which means $v\models\neg\Box\neg\alpha$. Fix $u$ with $R_Q(v,u)$. Symmetry of $R_M$ gives $R_M(v,w)$. The forcing condition applied to $R_M(v,w)$ and $R_Q(v,u)$ yields $R_M(u,w)$. Reflexivity of $R_Q$ and $w\models\alpha$ imply $w\not\models\neg\alpha$. Hence $u\not\models\Box\neg\alpha$. Since this holds for every $R_Q$-neighbor $u$ of $v$, we obtain $v\models\neg\Box\neg\alpha$. Thus $w\models\Box\Diamond\alpha$.
\end{proof}

The proof of B uses both symmetry of $R_M$ and the forcing condition. Classical Kripke semantics obtains the corresponding result from ordinary existential possibility.

\subsection{Maximal Consistent Pairs}

In orthologic, a consistent set need not have an extension in which, for every formula $\alpha$, either $\alpha$ or $\neg\alpha$ is present. The canonical construction represents truth and falsity in two coordinates. Every formula belongs to exactly one coordinate of a maximal pair, and membership in the right coordinate carries no identification with membership of the negated formula in the left coordinate. The extension argument below uses finitary derivations, CUT, and Zorn's lemma. This saturation in two coordinates matches the local interpretation of sequents with multiple conclusions.

\begin{definition}[Consistent Pair]
A pair $(A,B)$ of sets of formulas is \emph{consistent} if there are no finite sets $\Gamma\subseteq A$ and $\Delta\subseteq B$ such that $\Gamma\vdash_{\QMLE}\Delta$. It is \emph{maximal} if no proper coordinatewise extension is consistent.
\end{definition}

The first step is to extend every consistent pair to a maximal consistent pair.

\begin{lemma}[Extension of Consistent Pairs]
\label{lem:max-pair-extension}
Every consistent pair extends to a maximal consistent pair.
\end{lemma}

\begin{proof}
Fix a consistent pair $(A_0,B_0)$ and consider the set of its consistent coordinatewise extensions, ordered by coordinatewise inclusion. Let $\{(A_\lambda,B_\lambda)\}_{\lambda\in\Lambda}$ be a chain in this ordered set, and put $A=\bigcup_{\lambda\in\Lambda}A_\lambda$ and $B=\bigcup_{\lambda\in\Lambda}B_\lambda$. The pair $(A,B)$ extends $(A_0,B_0)$. If it were inconsistent, finite sets $\Gamma\subseteq A$ and $\Delta\subseteq B$ would satisfy $\Gamma\vdash_{\QMLE}\Delta$. Finiteness and the chain condition place every formula in $\Gamma\cup\Delta$ in one member $(A_\lambda,B_\lambda)$ of the chain, which contradicts the consistency of $(A_\lambda,B_\lambda)$. Thus every chain has an upper bound. Zorn's lemma gives a maximal consistent extension of $(A_0,B_0)$. Any proper consistent extension of this pair would also extend $(A_0,B_0)$, so the resulting pair is maximal among all consistent pairs.
\end{proof}

Maximality converts consistency into the partition and saturation properties required by the canonical construction.

\begin{lemma}[Properties of Maximal Pairs]
\label{lem:partition-saturation}
Let $(A,B)$ be a maximal consistent pair. Then:
\begin{enumerate}[(i)]
\item $A\cap B=\varnothing$ and $A\cup B$ is the set of all formulas.
\item if $\Gamma\subseteq A$ is finite and $\Gamma\vdash_{\QMLE}\Delta$, then $\Delta\cap A\neq\varnothing$.
\item if $A\vdash_{\QMLE}\alpha$, then $\alpha\in A$.
\item no formula $\alpha$ satisfies $\alpha,\neg\alpha\in A$.
\item for every $\alpha$, exactly one of $\Box\alpha$ and $\neg\Box\alpha$ belongs to $A$.
\end{enumerate}
\end{lemma}

\begin{proof}
If $\alpha\in A\cap B$, AX makes the pair inconsistent. Thus the coordinates are disjoint.

Suppose $\alpha\notin A\cup B$. If $(A\cup\{\alpha\},B)$ were inconsistent, every finite witness would use $\alpha$ on the left, since a witness omitting $\alpha$ would already make $(A,B)$ inconsistent. Thus finite $\Gamma_1\subseteq A$ and $\Delta_1\subseteq B$ would satisfy $\Gamma_1,\alpha\vdash_{\QMLE}\Delta_1$. Similarly, inconsistency of $(A,B\cup\{\alpha\})$ would yield finite $\Gamma_2\subseteq A$ and $\Delta_2\subseteq B$ satisfying $\Gamma_2\vdash_{\QMLE}\Delta_2,\alpha$. Applying CUT to the second sequent and then the first gives $\Gamma_1,\Gamma_2\vdash_{\QMLE}\Delta_1,\Delta_2$, contrary to the consistency of $(A,B)$. At least one coordinate can therefore be enlarged, and maximality places $\alpha$ in $A\cup B$. This proves (i).

For (ii), suppose $\Gamma\subseteq A$, $\Gamma\vdash_{\QMLE}\Delta$, and $\Delta\cap A=\varnothing$. By (i), $\Delta\subseteq B$, so the pair is inconsistent. Item (iii) is the special case of (ii) in which $\Delta$ has a single member.

For (iv), AX and $\neg$l give \(\alpha,\neg\alpha\vdash_{\QMLE}\varnothing.\)
Thus the two formulas cannot both belong to $A$.

For (v), MEM gives \(\varnothing\vdash_{\QMLE}\Box\alpha,\neg\Box\alpha.\)
Item (ii) places at least one conclusion in $A$, and item (iv) excludes both.
\end{proof}

The next consequence gives the form of double negation used in the canonical compatibility arguments.

\begin{lemma}[Double Negation in Maximal Pairs]
\label{lem:double-negation-pairs}
For every maximal consistent pair $(A,B)$,
\[
\alpha\in A\quad\Longleftrightarrow\quad\neg\neg\alpha\in A.
\]
\end{lemma}

\begin{proof}
We have $\alpha\vdash_{\QMLE}\neg\neg\alpha$ by AX and $\neg\neg$r. We also have $\neg\neg\alpha\vdash_{\QMLE}\alpha$ by AX and $\neg\neg$l. Apply Lemma~\ref{lem:partition-saturation}(iii) in both directions.
\end{proof}

\subsection{Canonical Structure}

Let $W^c$ be the set of left coordinates of maximal consistent pairs. The right coordinate is uniquely the complement of the left coordinate by Lemma~\ref{lem:partition-saturation}(i), so a world will be denoted by $u$, $v$, or $w$, as context requires.

Define
\begin{align*}
R_Q^c(u,v)
&\quad\Longleftrightarrow\quad
\forall\alpha\,\bigl(\alpha\in u\Longrightarrow\neg\alpha\notin v\bigr),\\
R_M^c(u,v)
&\quad\Longleftrightarrow\quad
\forall\alpha\,\bigl(\Box\alpha\in u\Longrightarrow\alpha\in v\bigr),\\
\rho^c(p)
&\quad:=\quad
\{u\in W^c\mid p\in u\}.
\end{align*}

\begin{remark}[Existence of Canonical Worlds]
The set $W^c$ is nonempty. By soundness, every derivable sequent is valid in a component model with one world, whereas the empty sequent is not valid there. Hence $\varnothing\nvdash_{\QMLE}\varnothing$, so $(\varnothing,\varnothing)$ is a consistent pair. Lemma~\ref{lem:max-pair-extension} supplies a maximal consistent pair and therefore a member of $W^c$.
\end{remark}

\begin{lemma}[Canonical Compatibility Relation]
\label{lem:canonical-rq}
The relation $R_Q^c$ is reflexive and symmetric.
\end{lemma}

\begin{proof}
Reflexivity follows from Lemma~\ref{lem:partition-saturation}(iv). Suppose $R_Q^c(u,v)$ and assume that symmetry fails. Then some $\beta$ satisfies $\beta\in v$ and $\neg\beta\in u$. Lemma~\ref{lem:double-negation-pairs} gives $\neg\neg\beta\in v$. Taking $\alpha=\neg\beta$ contradicts $R_Q^c(u,v)$.
\end{proof}

Reflexivity and symmetry provide the basic orthoframe structure. The next lemma supplies compatible witnesses when a negated formula is absent.

\begin{lemma}[Compatibility Extension]
\label{lem:compatibility-extension}
\[
\neg\alpha\notin u
\quad\Longrightarrow\quad
\exists v\in W^c\,
\bigl(R_Q^c(u,v)\text{ and }\alpha\in v\bigr).
\]
\end{lemma}

\begin{proof}
Consider the pair \(\left(\{\alpha\},\{\neg\beta\mid\beta\in u\}\right).\)
We first prove that it is consistent. An inconsistency witness would have a left side contained in $\{\alpha\}$ and a finite right side of the form $\{\neg\beta_1,\ldots,\neg\beta_n\}$, where every $\beta_i$ belongs to $u$ and $n$ may be zero. WKN adds $\alpha$ when the witness has an empty left side, so we obtain $\alpha\vdash_{\QMLE}\neg\beta_1,\ldots,\neg\beta_n$. The rule $\neg$r then gives $\neg\neg\beta_1,\ldots,\neg\neg\beta_n\vdash_{\QMLE}\neg\alpha$. Lemma~\ref{lem:double-negation-pairs} places every $\neg\neg\beta_i$ in $u$. When $n=0$, the antecedent of the last sequent is empty and is still a finite subset of $u$. Lemma~\ref{lem:partition-saturation}(iii) gives $\neg\alpha\in u$, contrary to the hypothesis.

Extend the pair to a maximal pair $(v,B)$ by Lemma~\ref{lem:max-pair-extension}. Then $\alpha\in v$. For every $\beta\in u$, the formula $\neg\beta$ belongs to $B$ and therefore not to $v$. This is exactly $R_Q^c(u,v)$.
\end{proof}

This witness property is sufficient to prove that the canonical atomic valuation is stable.

\begin{lemma}[Stability of the Canonical Valuation]
\label{lem:canonical-valuation-stable}
For every atomic formula $p$, the set $\rho^c(p)$ is $R_Q^c$-stable.
\end{lemma}

\begin{proof}
Let $X=\rho^c(p)$. The inclusion $X\subseteq X^{\perp\perp}$ holds in every orthoframe. Suppose $p\notin u$. Lemma~\ref{lem:double-negation-pairs} gives $\neg\neg p\notin u$. Apply Lemma~\ref{lem:compatibility-extension} to the formula $\neg p$. We obtain $v$ with $R_Q^c(u,v)$ and $\neg p\in v$.

Let $w$ satisfy $R_Q^c(v,w)$. If $p\in w$, Lemma~\ref{lem:double-negation-pairs} gives $\neg\neg p\in w$. Taking $\alpha=\neg p$ contradicts $R_Q^c(v,w)$. Therefore $p\notin w$ for every $R_Q^c$-neighbor $w$ of $v$, so $v\in X^{\perp}$. Since $R_Q^c(u,v)$, it follows that $u\notin X^{\perp\perp}$. Hence $X^{\perp\perp}\subseteq X$.
\end{proof}

We next verify that the canonical modal relation satisfies the forcing condition.

\begin{lemma}[Canonical Forcing]
\label{lem:canonical-forcing}
The relation $R_M^c$ satisfies the forcing condition relative to $R_Q^c$.
\end{lemma}

\begin{proof}
Assume $R_M^c(u,w)$ and $R_Q^c(u,v)$. We show $R_M^c(v,w)$. Let $\Box\alpha\in v$. If $\Box\alpha\notin u$, Lemma~\ref{lem:partition-saturation}(v) gives $\neg\Box\alpha\in u$. Lemma~\ref{lem:double-negation-pairs} gives $\neg\neg\Box\alpha\in v$. Taking the formula $\neg\Box\alpha$ contradicts $R_Q^c(u,v)$. Therefore $\Box\alpha\in u$, and $R_M^c(u,w)$ gives $\alpha\in w$.
\end{proof}

The forcing property alone does not provide witnesses for failure of a boxed formula. The following extension lemma supplies the required modal successor.

\begin{lemma}[Extension to a Modal Successor]
\label{lem:modal-successor-extension}
\[
\Box\alpha\notin u
\quad\Longrightarrow\quad
\exists v\in W^c\,
\bigl(R_M^c(u,v)\text{ and }\alpha\notin v\bigr).
\]
\end{lemma}

\begin{proof}
Put \(A_u:=\{\beta\mid\Box\beta\in u\}.\)
We claim that $(A_u,\{\alpha\})$ is a consistent pair. An inconsistency witness has a finite antecedent $\{\beta_1,\ldots,\beta_n\}\subseteq A_u$ and a succedent contained in $\{\alpha\}$. When the succedent is empty, WKN adds $\alpha$. Thus $\beta_1,\ldots,\beta_n\vdash_{\QMLE}\alpha$, where $n$ may be zero. The K rule gives $\Box\beta_1,\ldots,\Box\beta_n\vdash_{\QMLE}\Box\alpha$. Every formula in the antecedent belongs to $u$, and Lemma~\ref{lem:partition-saturation}(iii) gives $\Box\alpha\in u$, contrary to the hypothesis.

Extend $(A_u,\{\alpha\})$ to a maximal pair $(v,B)$. Then $A_u\subseteq v$, so $R_M^c(u,v)$, while $\alpha\in B$ implies $\alpha\notin v$.
\end{proof}

The canonical structure is now equipped with the relations and witness lemmas needed for the truth lemma. It remains to derive the equivalence properties of the modal relation from T, 4, and B.

\subsection{Canonical Equivalence Relation}

The three added axioms establish the frame properties of an equivalence relation in sequence. Specifically, reflexivity is derived from T, transitivity from 4, and symmetry from B.

\begin{lemma}[Canonical Reflexivity]
\label{lem:canonical-reflexivity}
$R_M^c$ is reflexive.
\end{lemma}

\begin{proof}
Let $\Box\alpha\in u$. Axiom T and Lemma~\ref{lem:partition-saturation}(iii) give $\alpha\in u$. Hence $R_M^c(u,u)$.
\end{proof}

\begin{lemma}[Canonical Transitivity]
\label{lem:canonical-transitivity}
$R_M^c$ is transitive.
\end{lemma}

\begin{proof}
Assume $R_M^c(u,v)$ and $R_M^c(v,w)$. Let $\Box\alpha\in u$. Axiom 4 gives $\Box\Box\alpha\in u$. Hence $\Box\alpha\in v$, and then $\alpha\in w$. Therefore $R_M^c(u,w)$.
\end{proof}

Symmetry, unlike reflexivity and transitivity, requires the interaction of B with canonical compatibility and the forcing condition.

\begin{lemma}[Canonical Symmetry]
\label{lem:canonical-symmetry}
$R_M^c$ is symmetric.
\end{lemma}

\begin{proof}
Assume $R_M^c(u,v)$. To prove $R_M^c(v,u)$, let $\Box\alpha\in v$. Suppose, towards a contradiction, that $\alpha\notin u$. Lemma~\ref{lem:double-negation-pairs} gives $\neg\neg\alpha\notin u$. Apply Lemma~\ref{lem:compatibility-extension} to $\neg\alpha$. Thus there is a world $w\in W^c$ such that \(R_Q^c(u,w)\text{ and }\neg\alpha\in w.\)
Lemma~\ref{lem:canonical-forcing} gives $R_M^c(w,v)$.

Axiom B, applied to $\neg\alpha$, yields \(\neg\alpha\vdash_{\QMLE}\Box\Diamond\neg\alpha.\)
Thus $\Box\Diamond\neg\alpha\in w$, and $R_M^c(w,v)$ gives \(\Diamond\neg\alpha\in v.\)
By abbreviation, this is \(\neg\Box\neg\neg\alpha\in v.\)
Also, $\alpha\vdash_{\QMLE}\neg\neg\alpha$, and K gives \(\Box\alpha\vdash_{\QMLE}\Box\neg\neg\alpha.\)
Since $\Box\alpha\in v$, we obtain $\Box\neg\neg\alpha\in v$. This contradicts Lemma~\ref{lem:partition-saturation}(iv). Hence $\alpha\in u$.
\end{proof}

The preceding lemmas place the canonical modal relation in the intended equivalence frame class.

\begin{corollary}[Canonical Frame Membership]
\label{cor:canonical-in-E}
The canonical structure
\[
\mathcal{M}^c=(W^c,R_Q^c,R_M^c,\rho^c)
\]
belongs to $\mathsf{E}$.
\end{corollary}

\begin{proof}
Use Lemmas~\ref{lem:canonical-rq}, \ref{lem:canonical-valuation-stable}, \ref{lem:canonical-forcing}, \ref{lem:canonical-reflexivity}, \ref{lem:canonical-transitivity}, and \ref{lem:canonical-symmetry}.
\end{proof}

\subsection{Canonical Completeness}

\begin{theorem}[Truth Lemma]
\label{thm:truth-lemma}
For every formula $\alpha$ and every canonical world $u$,
\[
u\models_{\mathcal{M}^c}\alpha
\quad\Longleftrightarrow\quad
\alpha\in u.
\]
\end{theorem}

\begin{proof}
We proceed by induction on $\alpha$.

\emph{Base case.}
If $\alpha$ is atomic, the claim follows from the definition of $\rho^c$.

\emph{Induction step.}
Assume that the claim holds for the immediate subformulas of $\alpha$.

\emph{Case $\alpha=\beta\wedge\gamma$.}
If $\beta\wedge\gamma\in u$, the derivations $\beta\wedge\gamma\vdash_{\QMLE}\beta$ and $\beta\wedge\gamma\vdash_{\QMLE}\gamma$ place both $\beta$ and $\gamma$ in $u$. Conversely, if $\beta,\gamma\in u$, the rule $\wedge$r applied to weakened identity sequents gives $\beta,\gamma\vdash_{\QMLE}\beta\wedge\gamma$, so $\beta\wedge\gamma\in u$. The induction hypothesis completes this case.

\emph{Case $\alpha=\neg\beta$.}
Suppose $\neg\beta\in u$. If $R_Q^c(u,v)$, then $\beta\notin v$. Otherwise, Lemma~\ref{lem:double-negation-pairs} would give $\neg\neg\beta\in v$, contrary to $R_Q^c(u,v)$ and $\neg\beta\in u$. The induction hypothesis gives $v\not\models\beta$. Hence $u\models\neg\beta$. Conversely, suppose $\neg\beta\notin u$. Lemma~\ref{lem:compatibility-extension} gives $v$ such that $R_Q^c(u,v)$ and $\beta\in v$. By the induction hypothesis, $v\models\beta$, so $u\not\models\neg\beta$.

\emph{Case $\alpha=\Box\beta$.}
Suppose $\Box\beta\in u$. If $R_M^c(u,v)$, then $\beta\in v$ by the definition of $R_M^c$. The induction hypothesis gives $v\models\beta$. Hence $u\models\Box\beta$. Conversely, suppose $\Box\beta\notin u$. Lemma~\ref{lem:modal-successor-extension} gives $v$ such that $R_M^c(u,v)$ and $\beta\notin v$. By the induction hypothesis, $v\not\models\beta$, so $u\not\models\Box\beta$.
\end{proof}
With the truth lemma established, failure of derivability can now be converted into a countermodel in the canonical frame.

\begin{theorem}[Canonical Completeness]
\label{thm:canonical-completeness}
For every set $\Gamma$ of formulas and every formula $\alpha$,
\[
\Gamma\models_{\mathsf{E}}\alpha
\quad\Longrightarrow\quad
\Gamma\vdash_{\QMLE}\alpha.
\]
\end{theorem}

\begin{proof}
Suppose $\Gamma\not\vdash_{\QMLE}\alpha$. We verify that $(\Gamma,\{\alpha\})$ is consistent. An inconsistency witness would consist of a finite $\Gamma_0\subseteq\Gamma$ and a succedent $\Delta\subseteq\{\alpha\}$ such that $\Gamma_0\vdash_{\QMLE}\Delta$. If $\Delta=\{\alpha\}$, the definition of derivability from an arbitrary set gives $\Gamma\vdash_{\QMLE}\alpha$. If $\Delta=\varnothing$, WKN gives $\Gamma_0\vdash_{\QMLE}\alpha$, and the same definition again gives $\Gamma\vdash_{\QMLE}\alpha$. Both cases contradict the hypothesis.

Extend $(\Gamma,\{\alpha\})$ to a maximal pair $(u,B)$. Then $\Gamma\subseteq u$ and $\alpha\in B$, so $\alpha\notin u$. By the truth lemma, $u\models\Gamma$ and $u\not\models\alpha$ in the canonical model. Corollary~\ref{cor:canonical-in-E} places that model in $\mathsf{E}$. Therefore $\Gamma\not\models_{\mathsf{E}}\alpha$.
\end{proof}

The final step combines canonical completeness with soundness and the coincidence of the three frame logics.

\begin{theorem}[Component Completeness]
\label{thm:component-completeness}
For every $\Gamma$ and $\alpha$,
\[
\Gamma\vdash_{\QMLE}\alpha
\quad\Longleftrightarrow\quad
\Gamma\models_{\mathsf{C}}\alpha.
\]
The same equivalence holds with $\mathsf{C}$ replaced by $\mathsf{E}$ or $\mathsf{U}$.
\end{theorem}

\begin{proof}
Soundness for $\mathsf{E}$ is Theorem~\ref{thm:soundness-eq}, and completeness is Theorem~\ref{thm:canonical-completeness}. The logic coincidence theorem transfers the equivalence to $\mathsf{C}$ and $\mathsf{U}$.
\end{proof}

\section{Algebraic Scope}
\label{sec:comparison}

The relational results are best understood through the algebra of stable sets and through the wider theory of modal semantics formulated with compatibility relations. This section places the forcing condition within those frameworks and then states its expressive boundary.

\subsection{Modal Semantics with Compatibility}

Compatibility--accessibility frames provide relational representations for complete lattices equipped with unary modalities, including complete ortholattices as a distinguished case \cite{Holliday2022}. The forcing condition is a special case of the interaction conditions available in that framework. Theorem~\ref{thm:component-invariance} gives a complete structural description of this case by showing that accessibility successors are constant on compatibility components. Theorem~\ref{thm:boxed-centrality} then places every boxed truth set in the component Boolean algebra.

Possibility semantics based on compatibility has also been used to develop modal orthologic for epistemic modalities \cite{HollidayMandelkern2024}. That work establishes algebraic and relational semantics for a modal orthologic consequence relation. The present component semantics adds the classification of accessibility relations satisfying the forcing condition, the connectedization obtained by adding one point, and the coincidence of the frame logics in Theorem~\ref{thm:logic-coincidence}.

\subsection{Modal Correspondence over Graphs}

Correspondence theory over graphs and polarities provides first-order analyses of modal reduction principles across non-distributive relational semantics \cite{ConradieEtAl2024}. This general setting includes principles associated with T, 4, and B. In the present orthoframe semantics, these principles are combined with the forcing condition and local validity for sequents with multiple conclusions. The resulting canonical relation is an equivalence relation satisfying the forcing condition, and Theorem~\ref{thm:logic-coincidence} transfers its logic for sequents with one conclusion to the exact component relation. The classical theory of equivalence frames and universal relations supplies the corresponding modal background \cite{BlackburnDeRijkeVenema2001}.

\subsection{Quantum Logic with Central Modalities}

Algebraic modal quantum logics study operators that assign to an orthomodular proposition the greatest central proposition below it \cite{DomenechFreytesDeRonde2009}. Modal enrichment has also been used to formulate contextuality results over orthomodular structures \cite{DomenechFreytesDeRonde2008}. The general component box constructed here returns the greatest component proposition below a stable set. Proposition~\ref{prop:component-center} shows that component propositions form a complete Boolean subalgebra of the center. Equality with the full center requires each compatibility component to have trivial center.

The hard superselection models in Section~\ref{sec:component-modality} meet this condition. Proposition~\ref{prop:hilbert-stable-product} gives an isomorphism between the stable set ortholattice and the product of the sector lattices. It also shows that the component Boolean algebra is the center of that product. In this model class, Theorem~\ref{thm:component-interior-cover} gives
\[
\Box_CX=\max\{Y\in Z(\St(W,R_Q))\mid Y\subseteq X\}
\]
and
\[
\Diamond_CX=\min\{Y\in Z(\St(W,R_Q))\mid X\subseteq Y\}.
\]

The present relational analysis contributes a classification of the accessibility relations satisfying the forcing condition, a logic coincidence theorem for component semantics, and a complete sequent calculus based on local validity. Together, these results explain the frame-theoretic source of the Boolean modal range and provide a proof theory for that semantics.

The hard superselection restriction remains essential. The full projective space of a Hilbert direct sum contains rays that cross the summands, and the corresponding Hilbert lattice has trivial center. The sector product arises from a state space and proposition lattice that already respect the superselection decomposition.

\subsection{Expressive Boundary}

Theorem~\ref{thm:component-invariance} gives a direct limit on interpretation. Within one compatibility component, a relation satisfying the forcing condition cannot distinguish initial states. On a connected projective Hilbert frame, it has a single successor set shared by every world. Exact unitary graphs and ordinary measurement updates whose action depends on the initial state therefore fall outside the present frame class in nondegenerate connected models.

The dual $\Diamond\alpha=\neg\Box\neg\alpha$ is determined by orthologic negation. For the component relation it gives the component cover and supports the B axiom. It does not generally express the existence of an accessible world satisfying $\alpha$. Applications involving outcome executability require a separate semantic operation or a weaker interaction condition.

Reset relations of the form $W\times D$, whose successor set does not depend on the initial state, remain natural examples on connected frames. A systematic treatment would require a multimodal language and explicit interaction principles between component and reset modalities. Such an extension lies beyond the present completeness theorem.

\section{Conclusion}

The forcing condition has a precise componentwise meaning. It makes modal successor sets invariant throughout each compatibility component and places every boxed proposition in the component Boolean algebra. The exact component relation realizes the associated interior and cover operations. In hard superselection models, these operations coincide with the greatest central lower approximation and the least central upper approximation.

The local interpretation of sequents with multiple conclusions provides the semantic form required by modal excluded middle on frames with several components. Maximal consistent pairs support a canonical model without imposing negation completeness on quantum propositions. The canonical accessibility relation satisfies the forcing condition, and the axioms T, 4, and B make it reflexive, transitive, and symmetric.

The main transfer result shows that exact component frames have the same logic for sequents with one conclusion as equivalence frames satisfying the forcing condition and connected frames carrying universal modal accessibility. The connectedization obtained by adding one point is the key step in this identification. Canonical completeness for equivalence frames satisfying the forcing condition therefore yields completeness for the component modality.

The resulting framework is a logic of modal certainty determined by compatibility components. Its scope includes central approximations and hard superselection semantics. Quantum dynamics whose action depends on the initial state requires a different interaction between compatibility and accessibility. Future work should examine algebraic representation beyond the case in which each component has trivial center, conservativity over the nonmodal calculus, and decidability of the resulting system.

\appendix
\section{Local Soundness for the Base Calculus}
\label{app:base-soundness}

This appendix presents the local soundness checks used in Theorem~\ref{thm:soundness-eq}. Let $\mathcal{M}$ be a quantum modal structure.

\begin{proposition}[Local Soundness for the Nonmodal Base Calculus]
Every nonmodal axiom and rule of the base calculus preserves local validity.
\end{proposition}

\begin{proof}
AX is immediate. WKN is immediate because adding premises narrows the worlds under consideration and adding conclusions enlarges the available alternatives.

For CUT, suppose both upper sequents are valid and let $w$ satisfy $\Gamma_1\cup\Gamma_2$. The first upper sequent gives a formula in $\Delta_1$ at $w$ or gives $\alpha$. In the first case the lower sequent holds. In the second case, the second upper sequent gives a formula in $\Delta_2$, so the lower sequent again holds.

The two conjunction left rules follow from the truth clause for conjunction. For conjunction right, let $w$ satisfy $\Gamma$. The first premise gives a member of $\Delta$ or $\alpha$, and the second gives a member of $\Delta$ or $\beta$. If either premise gives a member of $\Delta$, the lower sequent holds. Otherwise both $\alpha$ and $\beta$ hold, so $\alpha\wedge\beta$ holds.

For $\neg$l, suppose $w\models\neg\alpha$ and $w\models\Gamma$. The premise gives a member of $\Delta$ or $\alpha$. Reflexivity of $R_Q$ makes $\alpha$ incompatible with $\neg\alpha$ at the same world, so the latter alternative is impossible.

For $\neg$r, suppose $w\models\neg\Delta$. Let $v$ satisfy $R_Q(w,v)$. If $v\models\alpha$, validity of the premise gives some $\delta\in\Delta$ with $v\models\delta$. This contradicts $w\models\neg\delta$. Thus every $R_Q$-neighbor of $w$ fails $\alpha$, and $w\models\neg\alpha$.

The double negation rules are sound because every formula truth set is stable. Hence $\llbracket\alpha\rrbracket=\llbracket\neg\neg\alpha\rrbracket$.
\end{proof}

\section*{Funding}

This work was supported by Japan Society for the Promotion of Science (JSPS) KAKENHI [Grant Number JP24K03372].

\end{document}